\documentclass[10pt]{article}
\usepackage{bm}% bold math
\usepackage{verbatim}       % Defines \begin{comment}, \end{comment}

\usepackage{amscd}
\usepackage{amsthm}
\usepackage{amssymb}

\usepackage{amstext}
\usepackage{psfrag}
\usepackage{amsmath}
\usepackage{amsfonts}
\usepackage{units}
\usepackage{mathrsfs}
\newcommand{\bd}{\begin{definition}}                %inizia definizione
\newcommand{\ed}{\end{definition}}                  %fine definizione
\newcommand{\bc}{\begin{corollary}}                 %inizia corollario
\newcommand{\ec}{\end{corollary}}                   %fine corollario
\newcommand{\bl}{\begin{lemma}}                     %inizia lemma
\newcommand{\el}{\end{lemma}}                       %fine lemma
\newcommand{\bp}{\begin{proposition}}            %inizia proposizione
\newcommand{\ep}{\end{proposition}}                %fine proposizione
\newcommand{\bere}{\begin{remark}}                  %inizia osservazione
\newcommand{\ere}{\end{remark}}                     %fine oservazione

\newcommand{\bt}{\begin{theorem}}
\newcommand{\et}{\end{theorem}}

\newcommand{\be}{\begin{equation}}
\newcommand{\ee}{\end{equation}}

\newcommand{\bit}{\begin{itemize}}
\newcommand{\eit}{\end{itemize}}
\newtheorem{theorem}{Theorem}[section]
\newtheorem{corollary}[theorem]{Corollary}
\newtheorem{lemma}[theorem]{Lemma}
\newtheorem{proposition}[theorem]{Proposition}
\theoremstyle{definition}
\newtheorem{definition}[theorem]{Definition}
\theoremstyle{remark}
\newtheorem{remark}[theorem]{Remark}
\newtheorem{example}[theorem]{Example}

\begin{document}
%
%\DeclareGraphicsExtensions{.pdf}

%\title{Connection between Lorentzian distance and mechanical least action in spacetimes admitting a parallel null vector}

\title{Further observations on the definition of global hyperbolicity under low regularity}

\author{E. Minguzzi \footnote{Dipartimento di Matematica Applicata, Universit\`a degli Studi di Firenze,  Via
S. Marta 3,  I-50139 Firenze, Italy. E-mail:
ettore.minguzzi@unifi.it}}

%\pacs{04.20.Gz, 04.30.-w}

\date{February 2023}
\maketitle

\begin{abstract}
\noindent
%In some previous work we showed that the definition of global hyperbolicity for closed cone structure actually coincides with the definition we previously gave for  the more general framework of topological ordered spaces.
The definitions of global hyperbolicity for closed cone structures and topological preordered spaces are known to coincide.
In this work we clarify the connection with definitions of global hyperbolicity  proposed in recent literature on Lorentzian length spaces and Lorentzian optimal transport, suggesting possible corrections for the  terminology adopted in these works.
It is found that in Kunzinger-S\"amann's Lorentzian length spaces the definition of global hyperbolicity coincides with that valid for closed cone structures and, more generally, for topological preordered spaces: the causal relation is a closed order and the causally convex hull operation preserves compactness. In particular, it is independent of the metric, chronological relation or Lorentzian distance.

%, and showing that, once again, the working definition is really coincident with that previously proposed for the more general framework of topological order spaces.
%Some issues connected with the defini
%It is pointed out that some recent work in low regularity have adopted definition of global hyperbolicity that conflict with previous one. Once the working definik
%Prompted by some  recent work on Lorentzian length spaces and optimal transport of Lorentzian spacetimes where peculiar definitions of global hyperbolicity have been adopted, I show how, the actually working definition coincides with that proposed in 2013 for the more general framework of topological ordered spaces.
\end{abstract}

\section{Introduction}
Among the many causality properties that can be imposed on a Lorentzian spacetimes, global hyperbolicity is certainly one of the most useful. The PDE evolution of Cauchy data for the  Einstein's equations naturally lead to, potentially extendible, globally hyperbolic spacetimes. The belief that under physically reasonable conditions the globally hyperbolic spacetimes so obtained are inextendibile is known as the strong cosmic censorship.

Global hyperbolicity is also the strongest property in the causal ladder of spacetimes, and the evolution and refinement of its definition has reflected the progress of mathematical relativity in the last decades.

Recently, studies in low regularity Lorentzian geometry, by means of metric geometry, cone structures, length spaces, Lorentzian optimal trasport,
 have led to new investigations and adaptations of this property to more general frameworks.

Some years ago we started introducing elements of the theory of topological ordered spaces as developed by Nachbin \cite{nachbin65} in the study of the spacetime causal structure \cite{minguzzi09c}. The idea was to regard the spacetime structure as  a topological ordered space $(X,\mathscr{T}, J)$ endowed with a measure $\mu$. As argued  in \cite{minguzzi13e}, this type of framework could be sufficiently general to correctly describe a quantum spacetime theory, the manifold smoothness being expected to be lost at small length scales/high energies.

Nachbin theory was not sufficiently general though, as it was particularly lacking in connection with non-compact manifolds. In order to extend its range of application we obtained some results that could be applied to locally compact $\sigma$-compact spaces and hence to manifolds \cite{minguzzi11f,
%minguzzi11c,
%minguzzi12b,
minguzzi12d}. In this connection, we proposed a definition of global hyperbolicity that applies to closed ordered spaces \cite{minguzzi12d}.

It is the purpose of this work to show that this definition passes the test of time, as it is consistent with all the definitions of global hyperbolicity that have subsequently been proposed. This consistency was proved for closed cone structures in \cite{minguzzi17,minguzzi19c}, but recent work on Lorentzian length spaces and Lorentzian optimal transport, adopting new terminology, has brought us to reconsider this problem for these type of structures.
% with the goal of  in order to get a unifying picture.

We recall that a {\em topological preordered space} is a triple $(X,\mathscr{T},J)$ where $(X,\mathscr{T})$ is a topological space and $J$ is a reflexive and transitive relation ({\em preorder}). It is called an {\em order} if it is antisymmetric, in which case the triple is a {\em topological ordered space}. Several consistency conditions can be imposed between topology and order, in fact, as Nachbin showed \cite{nachbin65}, there is a beautiful topological theory for these spaces that is analogous to the usual topology. One of the most important conditions is that the preorder $J$, regarded as a subset of $X\times X$ be closed in the product topology $\mathscr{T}\times \mathscr{T}$, in which case we speak of {\em closed preordered space} (in some papers in topology this property is referred as  {\em continuity} of $J$). A weaker property is that of {\em semi-closedness} (or semi-continuity) namely $J$ is such that the sets $J^+(p):=\{q: (p,q)\in J\}$ and $J^-(p)=\{q: (q,p)\in J\}$ are closed for every $p$.

Observe that the antisymmetry of $J$ reads $\Delta=J\cap J^{-1}$ where $J^{-1}=\{(x,y): (y,x)\in J\}$ and $\Delta\subset X\times X$ is the diagonal. Thus for a closed order the diagonal $\Delta$ is closed, i.e.\   the  topology $\mathscr{T}$ is Hausdorff.
%Under this closedness condition  the topology $\mathscr{T}$ is Hausdorff as Hausdorffness  is equivalent to the fact that the diagonal $\Delta\subset X\times X$ is closed, this property following from the antisymmetry of the relation, which reads $\Delta=J\cap J^{-1}$ where $J^{-1}=\{(x,y): (y,x)\in J\}$.

The definition proposed in  \cite{minguzzi12d} was
\begin{definition} \label{viq}
A topological preordered space  $(X,\mathscr{T},J)$ is {\em globally hyperbolic} if
\begin{itemize}
\item[($\star$)] $J$ is a closed order and for every compact set $K$ the set $J^+(K)\cap J^-(K)$ is compact.
\end{itemize}
\end{definition}

Property ($\star$) consists of two conditions. The former might also be called {\em causal simplicity} while the latter reads: the operation of taking the causally convex hull preserves compactness. The weaker condition of {\em causality} is the request: $J$ is antisymmetric (i.e.\ an order).
Finally, it is possible to introduce an intermediate level between causal simplicity and causality, namely {\em stable causality} (which we shall not use): the smallest closed preorder containing $J$ is antisymmetric. This proves that portions of the causal ladder for spacetimes \cite{minguzzi18b} pass to the topological preordered space framework.

The property ($\star$) is important as for locally compact sigma-compact topologies it implies another very important property know as quasi-uniformizability, which essentially establishes the representability of the topological ordered space by continuous isotone functions and the fact that the space  can be Nachbin compactified \cite{minguzzi12d} (we shall not expand on these properties as they will not be used in what follows).

%The only undesirable feature of the previous definition is the closure condition on $J$ which one might want to drop.
%This is indeed possible, as we shall see.
%We shall return on this point later on.

\section{Global hyperbolicity on smooth manifolds}

Let us briefly recall the improvements on the definition of global hyperbolicity that took place along the years. We start with the regular setting, meaning with this $C^2$ (or $C^{1,1}$) metrics on smooth manifolds.

The traditional definition, as can be found in the oldest textbooks \cite{hawking73,beem96}, is {\em strong causality and compactness of the causal diamonds: $\forall p,q\in M$, $J^+(p)\cap J^{-}(q)$}. Bernal and S\'anchez proved that {\em  strong causality} could be weakened to {\em causality} \cite{bernal06b}, while in \cite{minguzzi19c} we proved that for physically reasonable spacetimes, i.e.\ with dimension larger than three which are non-compact or non-totally vicious, the condition of causality could be dropped altogether resulting in just {\em compactness of the causal diamonds}.

Other equivalent definitions were also obtained, for instance {\em non-total imprisonment and the causal diamonds are relatively compact} \cite{minguzzi08e} (a spacetime is non-total imprisoning if no inextendible causal curve is imprisoned in a compact set). This definition shows clearly that shrinking the cones does not spoil global hyperbolicity (as the family of causal curves get smaller, as do the causal diamonds)
and is particularly convenient for proving its stability under $C^0$ perturbations of the cones \cite{minguzzi11e,samann16,minguzzi17}. It also sets a balance between compact sets (and hence open sets) and causality, as the more the compact sets the harder for causal curves to be non-imprisoned, but the easier for causal diamonds to be relatively compact.

The first work in a low regularity setting was due to Fathi and Siconolfi \cite{fathi12,fathi15} who studied $C^0$ cone distributions. Their definition of global hyperbolicity was essentially the traditional one, but for strong causality that was strengthened to stable causality. Chru\'sciel and Grant \cite{chrusciel12} studied systematically the $C^0$ Lorentzian metric theory, pointing out that the equalities $\overline{ J^{\pm}(p)}=\overline{ I^\pm(p)}$ (no causal bubbles), $I\circ J\cup J\circ I\subset I$ (push up), do not hold at this level of regularity. The two pathologies were in fact one and the same as was later show in \cite[Thm.\ 2.8]{minguzzi17}, see also \cite[Thm.\ 2.12]{grant20}.
%They also observed that  many proofs on limit curve theorems  \cite{minguzzi07c} actually generalized with minimal modifications to this framework.
The $C^0$ Lorentzian geometry  approach was also developed  by Sbierski \cite{sbierski15} in his study on the $C^0$ inextendibility of Schwarzschild spacetime. S\"amann studied specifically global hyperbolicity and its many equivalent definitions in the same $C^0$ Lorentzian framework \cite{samann16}.

A more general point of view was taken by Bernard and Suhr, who studied closed cone structures \cite{bernard18}, thus reconsidering Fathi and Siconolfi's cone distribution approach. Subsequently, we explored quite in deep causality in closed cone structures and non-regular Lorentz-Finsler spaces \cite{minguzzi17}.

We recall that a {\em cone structure} on a smooth manifold is a multivalued map $x \mapsto C_x$, where $C_x \subset
T_xM\backslash 0$, is a closed sharp convex non-empty cone. It is a {\em closed cone structure} if $C=\cup_x C_x$ is a closed subbundle of the slit tangent bundle $TM\backslash 0$ (this is the terminology in \cite{minguzzi17}, Bernard and Suhr would call it {\em regular} closed cone structure). This is essentially an upper semi-continuity condition on the cone distribution \cite[Prop.\ 2.3]{minguzzi17}. It turns out that most of causality theory and the very causal ladder of spacetimes makes sense for closed cone structures \cite[Thm.\ 2.47]{minguzzi17}. However, in order to make sense of the chronological relation and some Lipschitzness condition on achronal hypersurfaces, one needs to work with a slightly more specialized object, namely a {\em proper cone structure} which is a closed cone structure which is proper, i.e.\ such $\textrm{Int} C\subset TM\backslash 0$ has non-empty fiber at each point of $M$. The $C^0$ metric Lorentzian framework is contained in the proper cone structure theory and hence in the theory of closed cone structures.

For a proper cone structure all the  definitions of global hyperbolicity  traditionally developed for smooth spacetimes have a straightforward analog and  remain equivalent among each other \cite{minguzzi17}. The relation of global hyperbolicity with causal simplicity and other causality properties also does not change. We might say that there are no surprises for proper cone structures, and no need to adjust the definitions \cite{minguzzi17,minguzzi19c}.

For closed cone structures one has to be  careful. The generalization of the traditional definition which makes use of causal diamonds does not work in the sense that it is not equivalent to other desirable properties, such as the existence of a Cauchy hypersurface \cite[Example 2.6]{minguzzi17}.
Nevertheless, the following properties are equivalent  \cite{bernard18} \cite[Thm.\ 2.39, 2.45]{minguzzi17}
\begin{itemize}
\item[($\star$)] $J$ is a closed order and for every compact set $K$, $J^+(K)\cap J^-(K)$ is compact.
\item[($*$)] Causality and for every compact sets $K_1$ and $K_2$, the `causal emerald' $J^+(K_1)\cap J^-(K_2)$ is compact.
\item[(2)] Non-imprisonment and the causally convex hulls of relatively compact sets are relatively compact.
\item[(3)] Existence of a Cauchy hypersurface.
\item[(4)] Existence of a Cauchy time function.
\end{itemize}
and hence any of them can provide the correct definition of global hyperbolicity in this framework. Observe that ($\star$) coincides with the definition
 \ref{viq} we gave in the context of topological ordered spaces, while property ($*$) was introduced by Bernard and Suhr in \cite{bernard18}.
In the recent work \cite{minguzzi19c} we established that for closed cone structures ($\star$) and ($*$) can be improved as follows (in the smooth setting this result had been already obtained in \cite{minguzzi11e})
\begin{itemize}
\item[(1)] $J$ is an order (causality) and for every compact set $K$, $J^+(K)\cap J^-(K)$ is compact.
\end{itemize}
%Unfortunately, our result was framed in the context of closed cone structures and could not be directly applied to Lorentzian length spaces. It would have possibly dispelled some confusion that was generated in the literature.

%Let us now come to recent work on Lorentzian length spa

\section{Global hyperbolicity for topological ordered \\ spaces}

The  most general setting for studying causality is that of topological preordered spaces.  Any framework for spacetime, including those not using a notion of smooth manifold, such as that of Lorentzian length spaces, can be seen as an instance of this general framework (sometimes with some caveats, see below).

In the introduction we provided definitions for casuality, stable causality, causal simplicity and global hyperbolicity, i.e.\ property ($\star$).

%The last property will be studied in some detail below. Here we want to prove that the causal ladder really passes to the framework of topological ordered spaces. In a previous work we showed that it is preserved at the level of closed cone structures \cite{minguzzi17}, and hence that it does not require for its formulation a notion of  chronological relation. Here we shown that the smoothness of the manifold, or the fact that the causal relation $J$ is constructed out of some classes of (causal) curves is also not required.
%
%Let $(M,\mathscr{T}, J)$ be a topological preordered space. We denote with $K$ the smallest closed preorder such that $J\subset K$. It exists because the family $\mathcal{K}$ of closed preorders that contain $J$ is non-empty, as it includes $M\times M$. By taking the intersection $K=\bigcap \mathcal{K}$ of that family we obtain the desired smallest preorder.
%
%\begin{definition}
%The topological preordered space $(M,\mathscr{T}, J)$ is
%\begin{itemize}
%\item stably causal: if $K$ is antisymmetric,
%\item non-total imprisoning: if for every compact set $C$ the topological preordered space $(C,\mathscr{T}\vert_C, J\cap (C\times C) )$ is stably causal,
%\end{itemize}
%\end{definition}
%
%
%
%
%\subsection{Global hyperbolicity}

We  observe that in ($\star$) the replacement of the compact set $K$ with two sets $K_1$, $K_2$, does not change the property as we have

\begin{proposition} \label{biw}
For a closed preordered space the properties
\begin{itemize}
\item[(i)] for every compact subset $K$, $J^+(K)\cap J^-(K)$ is compact,
\item[(ii)] for any two compact subsets $K_1$, $K_2$,   $J^+(K_1)\cap J^-(K_2)$ is compact,
\end{itemize}
are equivalent.
\end{proposition}

\begin{proof}
One direction is obvious setting $K=K_1=K_2$. For the other direction, it is well known that in a closed preordered space if $K$ is compact, $J^{\pm}(K)$, is closed \cite[p.\ 44]{nachbin65} \cite[Prop.\ 2.2]{minguzzi11f}. Thus $J^+(K_1)\cap J^-(K_2)$ is a closed subset of the compact set $J^+(K)\cap J^-(K)$, where $K=K_1\cup K_2$, hence compact.
\end{proof}

\begin{proposition}\label{bir}
Let $(M,\mathscr{T}, J)$ be a topological preordered space. Assume that the topology is Hausdorff and ``first countable or locally compact''.  The property
\begin{itemize}
\item[($\ddagger$)] for every compact sets $K_1,K_2$, $J^+(K_1)\cap J^-(K_2)$ is compact,
\end{itemize}
implies
\begin{itemize}
\item[($\sharp$)] $J$ is closed in the product topology.
\end{itemize}
\end{proposition}

We recall that under the Hausdorff condition if every point admits a compact neighborhood then every point admits a basis of compact neighborhoods. Thus there is not ambiguity in what we mean by local compactness. The proof is similar to \cite[Thm.\ 2.38]{minguzzi17}.

\begin{proof}
Let us give the proof in the `first countable' case.
First we prove that $J$ is semi-closed. Let $q\in \overline{J^+(p)}$ then there is a sequence $q_k\to q$, $p\le q_k$. The set $K=\{q,q_1, q_2, \cdots\}$ is compact, thus as $\{p\}$ is compact, $B=J^+(p)\cap J^-(K)$ is compact hence closed. But $q_k\in B$, thus $q\in B$ which implies $p\le q$. By the arbitrariness of $q$, $\overline{J^+(p)}=J^+(p)$. The fact that  $J^-(p)$ is closed is proved analogously.

Now, let  $(p,q)\in \overline{J}$, then we can find $(p_k,q_k)\in J$, $(p_k,q_k) \to (p,q)$. Let us consider the compact sets $K^p_n=\{p,p_n, p_{n+1}, \cdots\}$, $K^q_n=\{q,q_n, q_{n+1}, \cdots\}$, then $J^+(K^p_n)\cap J^-(K^q_n)$ is non-empty (as it contains $p_k$ and $q_k$ for every $k\ge n$) and compact. By the finite intersection property, there is $r\in M$ such that $r\in J^+(K^p_n)\cap J^-(K^q_n)$ for every $n$. By the semi-closure of $J$, $p \le r \le q$, thus $p\le q$.

In the locally compact case the proof is analogous, just let $K$ be a compact neighborhood in the first part, and let $K^p_\alpha, K^q_\beta$ be generic  compact neighborhoods of $p$ and $q$ in the second part.
\end{proof}

We are ready to prove that the equivalence between ($\star$) and ($*$), already proved in the context of closed ordered spaces \cite[Thm.\ 2.39]{minguzzi17}, actually holds in general for topological preordered spaces (in that result the Hausdorff property for the topology was contained in the manifold condition on  $M$).

\begin{theorem} \label{ppr}
Let $(M,\mathscr{T}, J)$ be a topological preordered space such that the topology is ``first countable or locally compact''.   The following properties are equivalent
\begin{itemize}
\item[($\star$)] $J$ is a closed order and for every compact set $K$, $J^+(K)\cap J^-(K)$ is compact,
\item[($\tilde{*}$)]  $J$ is antisymmetric, the topology is Hausdorff and for every compact sets $K_1,K_2$, $J^+(K_1)\cap J^-(K_2)$ is compact.
\end{itemize}
\end{theorem}

\begin{proof}
The direction ($\star$) $\Rightarrow$ ($\tilde{*}$) is Prop.\ \ref{biw}, noting that, as previously mentioned, a closed ordered space has Hausdorff topology. The direction ($\tilde{*}$) $\Rightarrow$ ($\star$) is Prop.\ \ref{bir}.
\end{proof}

\begin{example}
Let $(M,\mathscr{T}, J)$ be a topological ordered space. Consider the properties
\begin{itemize}
\item[($\dagger$)] for every compact set $K$, $J^+(K)\cap J^-(K)$ is closed,
\item[($\sharp$)] $J$ is closed in the product topology.
\end{itemize}
Is it true that the former implies the latter? The other direction is a well known as the closedness of $J$ implies the closedness
of $J^\pm(K)$, see \cite[p. 44]{nachbin65}.

We know that the result holds true in the smooth Lorentzian setting \cite[Lemma 2.1]{minguzzi11e} and, more generally, for closed cone structures \cite{minguzzi19c}. However, it does not pass to topological ordered spaces, not even under good properties for the topology (e.g.\ metrizable).

We can give the following minimal couterexample. Let $M=\{p,q,q_1,q_2, \cdots\}$, $M \subset \mathbb{R}$, where $q=0$, $q_n=1/n$, $p=-1$, and where the topology is the induced topology. Moreover, define $\le$ as follows: set $p\le q_k\nleq p $ for every $k$, $p, q_k\nleq q\nleq p, q_k$, for every $k$, and $q_j\nleq q_k$ for $j\ne k$. Observe that $q_k\to q$.
In this example every compact set is such that $J^+(K)\cap J^-(K)=K$ which is compact hence closed. Clearly, $J$ is not closed (not even semiclosed).
%If $K$ is a compact set that does not contain $q$ then $K$ has finite cardinality and $J^+(K)\cap J^-(K)=K$ which is compact hence closed.
%Any compact set has infinite cardinality if and only if it contains $q$, and in this case the causally convex hull is again the set itself which is compact hence closed. Clearly, $J$ is not closed (not even semiclosed).
\end{example}

As a preliminary result for the next section, it will be useful to recall that non-total imprisonment in the smooth setting is  defined as follows: there are no inextendible causal curves imprisoned in a compact set \cite{minguzzi18b}. The definition remains valid for closed cone structures \cite[Def.\ 2.10]{minguzzi17}. For them it is also true that a continuous causal curve is inextendible if and only if it has infinite $h$-arc length where $h$ is a complete Riemannian metric \cite[Cor.\ 2.1]{minguzzi17}. Finally, for this structure an easy application of the limit curve Lemma \cite[Lemma 2.1]{minguzzi17}  gives a result which is also valid in the smooth setting and in  the $C^0$ Lorentzian theory \cite{samann16}\cite[Thm.\ 4.39]{minguzzi18b}

\begin{proposition} \label{biq}
A closed cone structure $(M,C)$ is non-total imprisoning if and only if for an auxiliary (and hence for every) Riemannian metric $h$ on $M$ and every compact set $K$ there is a constant $c(K)>0$ such that all the continuous causal curves with image in $K$ have $h$-arc length smaller than $c(K)$.
\end{proposition}

\subsection{Lorentzian length spaces}
Let us discuss Lorentzian (pre-)length spaces in the version by Kunzinger and S\"amann \cite{kunzinger18}. We shall not recall all definitions, referring to \cite{kunzinger18} for details. A pre-length space $(M, \rho, I,J, d)$ is a quintuple given by a Kronheimer and Penrose's causal space $(M,I,J)$, a metric $\rho$, and a lower semi-continuous Lorentzian distance $d$, satisfying some compatibility conditions.

Proposition \ref{biq} has an analog in the theory by Kunzinger and S\"amann.  The two properties
\begin{itemize}
\item[($\alpha$)] there is no  inextendible causal curve imprisoned in compact set, and
\item[($\beta$)] the causal curves with image in a compact set have bounded length, (the length is that induced from $\rho$)
\end{itemize}
are equivalent for locally causally closed $\rho$-compatible Lorentzian pre-length space  \cite[Lemma 3.12, Cor.\ 3.15]{kunzinger18} hence  for Lorentzian length spaces.

They defined non-total imprisonment in their framework as property $(\beta)$, thus it depends on $\rho$. %Later we shall refer to this choice as KS-non-total imprisonment.
 Actually, the distance $\rho$ also appears in ($\alpha$) as it is present in  the Lipschitz condition that they impose on  their causal curves \cite[Def.\ 2.18]{kunzinger18}. The distance $\rho$ seems to be essential for their definitions. We observe that on manifolds and over  compact subsets distances associated to Riemannian metrics are all Lipschitz equivalent, however, their Lorentzian length spaces $M$ are not manifolds.

This situation clearly complicates causality in the Kunzinger-S\"amann theory as global hyperbolicity is obtained from ($\beta$) by adding additional conditions. It seems that causality not only depends on the causal relation $J$ and on the topology, but also on the metric $\rho$.

They
 defined global hyperbolicity through the property\footnote{It was claimed in \cite{ake20} that `non-total imprisonment' can be weakened to `causality' since the proof of  \cite[Thm.\ 3.26(v)]{kunzinger18} would not use `non-total imprisonment'. Actually, this is not correct as it uses it in applying \cite[Thm.\ 3.14]{kunzinger18} and the uniform bound on lengths implied by non-total imprisonment \cite[Def.\ 2.35]{kunzinger18}. The fact that causality and compactness of causal diamonds  implies ($\beta$) is not proved in \cite{ake20}, where the authors keep using their Def.\ 3.1 of global hyperbolicity and hence property ($\beta$)  in their arguments.  Still their claim that the assumption in global hyperbolicity can be weakened to causality is  correct under minimal assumptions on the Lorentzian length space, see  our last Corollary \ref{jje}.}
\begin{itemize}
\item[(a)] property ($\beta$) and for every $p,q\in M$, the `causal diamonds' $J^+(p)\cap J^-(q)$ are compact,
\end{itemize}
%Actually, as observed in \cite{ake20}, the proof of  \cite[Theorem 3.26(v)]{kunzinger18} shows that non-total imprisonment can be weakened to causality, so that it is really coincident with Bernal and S\'anchez version for smooth spacetimes.

Subsequently, in their paper on optimal transport over smooth Lorentzian manifolds, Mondino and Suhr \cite{mondino18} adopted the same definition
but realized than in order to develop their theory they needed a stronger property
\begin{itemize}
\item[(b)] property ($\beta$) and for every compact subsets $K_1,K_2$ the `causal emeralds' $J^+(K_1)\cap J^-(K_2)$ are compact.
\end{itemize}
which they called  `$\mathcal{K}$-global hyperbolicity'.

%Unfortunately, they called the first property `global hyperbolicity' and the second property `$K$-global hyperbolicity'. However, it was known that (a) could not be the right definition of global hyperbolicity under low regularity, for instance, as observed above, it does not work for closed cone structures. It does not surprise that in their results they end up using property (b), which should more simply be called `global hyperbolicity'. T
This terminology was adopted in Cavalletti and Modino work on optimal transport over  Lorentzian length spaces \cite{cavalletti20}, and in posterior works using the same framework \cite{braun22,braun23}.

We have the following result (this is stronger than
%\footnote{It seems that $\rho$-compatibility is missing in the assumptions of this lemma but used in step 2.}
\cite[Lemma 1.5]{cavalletti20})

\begin{proposition} \label{vic}
For a Lorentzian pre-length space such that every point admits a timelike curve passing through it (e.g.\ localizable ones and hence Lorentzian length spaces \cite{kunzinger18}),  properties (a) and (b) actually coincide.
\end{proposition}

Actually, the proof does not use the non-total imprisoning property ($\beta$), nor causality.
It is worth recalling that localizable Lorentzian pre-length spaces are $\rho$-compatible \cite[Def.\ 3.18]{kunzinger18}.

\begin{proof}
The direction (b) $\Rightarrow$ (a) is clear. For the direction (a) $\Rightarrow$ (b) one first proves that $J$ is closed in the usual way \cite[Thm.\ 4.12]{minguzzi18b}. Then the  proof is word by word that given in \cite[Prop.\ 2.3]{minguzzi19c} or \cite[Prop.\ 2.21]{minguzzi17}.
\end{proof}

Let us denote with $\mathscr{T}$ the topology of the pre-length space, that is, that induced by $\rho$. We recall that a causally path connected Lorentzian pre-length space has the following property \cite[Lemma 3.3]{kunzinger18}: causality (defined by the antisymmetry of $J$) holds if and only if there are no closed causal curves.

\begin{theorem} \label{iwe}
For a Lorentzian pre-length space which is causally path connected  and $\rho$-compatible  (hence for Lorentzian length spaces) the following properties are equivalent
\begin{itemize}
\item $(M,\mathscr{T}, J)$ satisfies ($\star$) (i.e. global hyperbolicity)
\item property (b) (i.e.\ so called $\mathcal{K}$-global hyperbolicity).
\end{itemize}
In particular, the second property does not depend on $\rho$ (or $I$, or $d$) just on its induced topology.
\end{theorem}

\begin{proof}
Assume ($\star$). By Thm.\ \ref{ppr} for every compact sets $K_1,K_2$, $J^+(K_1)\cap J^-(K_2)$ is compact.
Suppose that ($\beta$) does not hold then we can find a compact set $C$ and causal curves $\gamma_k: [0,L_k]\to C$ parametrized with $\rho$-arc length that are thus 1-Lipschitz and so equi-Lipschitz. By the Arzel\'a-Ascoli theorem a subsequence converges uniformly on compact subsets to a 1-Lipschitz curve $\gamma:[0,\infty) \to C$. Since $J$ is closed (hence locally causally closed), $\gamma$ is causal. By $\rho$-compatibility $\gamma$ is inextendible (the proof goes as in the last paragraph of the proof of \cite[Thm.\ 3.14]{kunzinger18}).  The inextendible curve $\gamma: [0,\infty)\to C$ imprisoned in a compact set $C$ cannot accumulate on just one point \cite[Lemma 3.12]{kunzinger18}, thus there must be sequences $s_k, t_k \to \infty$ such that $\lim_k \gamma(s_k)=p$, $\lim_k \gamma(t_k)=q$, $p,q\in C$, $p\ne q$.  Passing to subsequences if necessary, we can assume that $s_k<t_k < s_{k+1}$, thus $\gamma(s_k)\le \gamma(t_k)\le \gamma(s_{k+1})$, and taking the limit $(p,q), (q,p)\in \bar{J}$.
 But by  ($\star$) $J$ is closed, thus antisymmetry is violated, a contradiction. \\

%We need only to prove ($\alpha$) as under the assumptions it is equivalent to ($\beta$). Suppose ($\alpha$) does not hold then the inextendible curve $\gamma: [0,b)\to C$ imprisoned in a compact set $C$ cannot accumulate on just one point \cite[Lemma 3.12]{kunzinger18}, thus there must be sequences $s_k, t_k \to b$ such that $\lim_k \gamma(s_k)=p$, $\lim_k \gamma(t_k)=q$, $p,q\in C$, $p\ne q$.  Passing to subsequences if necessary, we can assume that $s_k<t_k < s_{k+1}$, thus $\gamma(s_k)\le \gamma(t_k)\le \gamma(s_{k+1})$, and taking the limit $(p,q), (q,p)\in \bar{J}$.
 %But by  ($\star$) $J$ is closed, thus antisymmetry is violated, a contradiction.

Conversely, assume (b), then by ($\beta$) there are no closed causal curves (as their image is compact)  hence $J$ is antisymmetric. By Thm.\ \ref{ppr} ($\star$) follows.
\end{proof}

The following result shows that under minimal conditions on the Lorentzian length space we can rescue a certain claim on the equivalence of two definitions of global hyperbolicity stated in \cite[Sec.\ 3]{ake20}.

\begin{corollary} \label{jje}
For a Lorentzian pre-length space which is causally path connected, $\rho$-compatible (e.g.\ a Lorentzian length space) and such that through each point passes a timelike curve, the properties (global hyperbolicity) (a), (b), and ($\star$), are all equivalent, and they are also equivalent to: causality and the causal diamonds are compact.
\end{corollary}

\begin{proof}
The first statement follows from Prop.\ \ref{vic} and Thm.\ \ref{iwe}. The last property is clearly implied by ($\star$). For the converse, assume causality and that the causal diamonds are compact.  First  $J$ is closed by the usual argument \cite[Thm.\ 4.12]{minguzzi18b}, hence $J$ is a closed order, and again by the argument in \cite[Prop.\ 2.3]{minguzzi19c} ($\star$) holds.
\end{proof}

\section{Conclusions}

We recalled the definition of global hyperbolicity for topological ordered spaces, the equivalent formulations for closed/proper cone structures, and explored some variations in the broader framework of topological preordered spaces (Thm.\ \ref{ppr}).

When it comes to Lorentzian (pre-)length spaces a la Kunzinger-S\"amann $(M, \rho, I,J, d)$, in my opinion, the terminology for property (b) introduced in the literature ($\mathcal{K}$-global hyperbolicity) should be rectified, as (b) could  be simply called `global hyperbolicity' (and (a) something like weak global hyperbolicity), as Lorentzian length spaces are special types of topological ordered spaces.
As shown in Thm.\ \ref{iwe}, property (b) coincides with the definition of global hyperbolicity previously introduced for these more general type of structures and as such it is, quite interestingly, independent of $\rho$, $I$ and $d$.
%
%Let $(M,\mathscr{T}, J)$ be a topological ordered space. We say that it is strongly causal if for every $x$ and every neighborhood $U\ni x$ we can find an open set $V\subset U$ such that $J^+(V)\cap J^{-}(V)\subset U$.
%\begin{proposition}
%Let $(M,\mathscr{T}, J)$ be a strongly causal topological ordered space. Suppose that the topology is first countable and Hausdorff. If $J^+(K)\cap J^-(K)$ is compact for every compact $K$, then $(M,\mathscr{T}, J)$  is strongly causal.
%\end{proposition}
%
%\begin{proof}
%Suppose that strong causality does not hold at $x\in M$, then there is $U$ and sequences $p_k,q_k$, $p_k\le q_k$, such that $p_k,q_k\to x$, but $J^+(p_k)\cap J^-(q_k)\cap U^C=\emptyset$. Now, consider the compact set $K_n=\{x,p_n,q_n,p_{n+1},q_{n+1},\cdots\}$, then $J^+(K_n)\cap J^-(K_n)\cap U^C$ is non-empty and compact. Since $K_{n+1}\subset K_n$, there is a point $r\in  \cap_n [J^+(K_n)\cap J^-(K_n)]\cap U^C$. We can pass to a subsequence so that $r \in J^+(p_n)$ for every $n$. As $r\in U^C$ we have $r\ne x$. Let $C=\{x,r\}$, we know that $J^+(C)\cap J^-(C)$ is compact, but ...........
%
%
%
%
%By the same argument there is a point $s$ such that $s\in \cap_n [J^+(K_n)\cap J^-(K_n)]\cap U^C$ and we can  pass to a further subsequence so that $s \in J^-(q_n)$ for every $n$.
%
%
%\end{proof}

%This fact might not be self evident but it follows from the following generalization of \cite[Thm.\ 2.10]{minguzzi19c}

\section*{Acknowledgments}
I thank Clemens S\"amann for  useful suggestions.
%, particularly in connection with Corollary \ref{jje}.
%This work has been partially supported by GNFM of INDAM and by FQXi.

%\bibliography{../../bibliografie/simultaneity,../../bibliografie/libri,../../bibliografie/miei,../../bibliografie/mieiPrep,../../bibliografie/mieiProc}
%\bibliographystyle{plain}

\end{document}